\documentclass[runningheads]{llncs}
\usepackage{graphicx}
\usepackage{amsmath,amssymb}
\usepackage[ruled, linesnumbered]{algorithm2e}
\usepackage{subcaption}
\usepackage{float}

\begin{document}
\title{The Batched Set Cover Problem}
%
%
\author{Juan C. Mart\'{i}nez Mori \inst{1} \and
Samitha Samaranayake \inst{2}}
\authorrunning{Mart\'{i}nez Mori, J. C. and Samaranayake, S.}
%
\institute{
Systems, Cornell University \\ 
\email{jm2638@cornell.edu} 
\and
School of Civil and Environmental Engineering, Cornell University \\
\email{samitha@cornell.edu}
}
\maketitle              
\begin{abstract}

We introduce the \textit{batched set cover problem}, which is a generalization of the online set cover problem. In this problem, the elements of the ground set that need to be covered arrive in batches. Our main technical contribution is a tight $\Omega(H_{m - 2^z + 1})$ lower bound on the competitive ratio of any fractional batched algorithm given an adversary that is required to produce batches of VC-dimension at least $z$, for some $z \in \mathbb{N}^0$. This restriction on the adversary is motivated by the fact that, in some real world applications, decisions are made after collecting batches of data of non-trivial VC-dimension. In particular, ridesharing systems rely on the batch assignment of trip requests to vehicles, and some related problems such as that of optimal congregation points for passenger pickups and dropoffs can be modeled as a batched set cover problem with VC-dimension greater than or equal to two. Furthermore, we note that while any online algorithm may be used to solve the batched set cover problem by artificially sequencing the elements in a batch, this procedure may neglect the rich information encoded in the complex interactions between the elements of a batch and the sets that contain them. Therefore, we propose a minor modification to an online algorithm found in \cite{buchbinder2009design} to obtain an algorithm that attempts to exploit such information. Unfortunately, we are unable to improve its analysis in a way that reflects this intuition. However, we present computational experiments that provide empirical evidence of a constant factor improvement in the competitive ratio. To the best of our knowledge, we are the first to use the VC-dimension in the context of online (batched) covering problems.

\end{abstract}

\section{Introduction}
\subsection{Background}
Let $X=\{1, \cdots, n\}$ be a ground set of $n$ elements and $\mathcal{S} = \{S_1, \cdots, S_m \}$ be a collection of $m$ subsets of $X$. A \textit{set cover} is a sub-collection of $\mathcal{S}$ such that its union is $X$. The \textit{set cover problem} is to find a minimum cardinality set cover of $X$. It is a classic NP-hard problem that is also hard to approximate to within a factor of $(1 - \alpha) \ln n $ in polynomial time for any $\alpha > 0$ \cite{feige1998threshold,irit2014analytical}.

In the online setting \cite{alon2003online,buchbinder2005online,buchbinder2009design}, the members of $\mathcal{S}$ are identified a priori, but the elements of the ground set that need to be covered, along with their respective set membership, are revealed sequentially. More precisely, the \textit{online set cover problem} consists of a game between an algorithm and an oblivious adversary; one which knows the algorithm but not the realization of any random choices\footnote{If the algorithm is deterministic, an oblivious adversary is equivalent to an adaptive adversary; one which makes requests adaptively in response to the algorithm \cite{ben1994power}.}. The adversary produces, in advance, a sequence $\sigma = \sigma_1, \sigma_2, \cdots$ of elements of $X$, which it reveals to the algorithm one at a time. Upon the arrival of an element, the algorithm must either conclude that the element is already in the set cover, or irrevocably extend the set cover with a member of $\mathcal{S}$ containing the element.

Alon et al. \cite{alon2003online} gave a deterministic $O(\log m \log n)$-competitive algorithm for the online set cover problem and a nearly matching lower bound for any online algorithm. Buchbinder and Naor \cite{buchbinder2005online,buchbinder2009design} later proposed a general scheme for the design and analysis of online algorithms, namely the primal-dual method\footnote{This scheme first appeared in the context of approximation algorithms \cite{williamson2011approximation}.}, and used it to obtain new algorithms for the online set cover problem. Their algorithms generally consist of two phases: \textit{i}) a deterministic $O(\log m)$ primal-dual subroutine for the fractional online set cover problem, which is optimal up to constant terms, and \textit{ii}) a randomized rounding procedure whose expected cost is $O(\log n)$ times the cost of the fractional solution, ultimately producing randomized $O(\log m \log n)$-competitive algorithms. The rounding procedure can be derandomized, producing deterministic $O(\log m \log n)$-competitive algorithms.

\subsection{Contributions}
Herein we introduce the \textit{batched set cover problem}, which is a generalization of the online set cover problem. However, as in \cite{buchbinder2005online,buchbinder2009design}, our focus is on its fractional counterpart; this corresponds to phase \textit{i}) of the primal-dual scheme. We immediately recover the integral case through the rounding procedures in phase \textit{ii}), which we leave untouched. In essence, the batched set cover problem differs from the online set cover problem in that the adversary produces a sequence of batches of elements of $X$. Thus, the online set cover problem is a special case of the batched set cover problem in which each element revealed by the adversary is its own batch. 

Note that the problem we consider is distinct from the capacitated online set cover problem with set requests treated by Bhawalkar et al. \cite{bhawalkar2014online}. They argue that the uncapacitated problem is not meaningful because the elements in a batch can be thought of as arriving sequentially, whereas we argue that this is not always the case. Our main technical contribution is a tight lower bound on the competitive ratio of any fractional batched algorithm given a parametrized restriction on the adversary. Specifically, if we consider adversaries that are required to produce batches of Vapnik Chervonenkis (VC)-dimension \cite{vapnik2015uniform} at least $z$, for some $z \in \mathbb{N}^0$, any fractional batched algorithm is $\Omega(H_{m - 2^z + 1})$-competitive. For $z > 0$, this bound is more generous (to the algorithm) than the $\Omega(H_m)$ bound of the online setting \cite{buchbinder2005online,buchbinder2009design}, which we recover when $z = 0$. 

In addition, we propose a minor modification to an online algorithm found in \cite{buchbinder2009design} to obtain a dedicated batched algorithm. The main idea is the simultaneous update of the dual variables that correspond to unsatisfied primal constraints, which is reminiscent of a primal-dual algorithm in \cite{williamson1995primaldual} for the generalized Steiner tree problem. Unfortunately, we are unable to analyze this algorithm in a way that exhibits the effects of the more generous (parametrized) bound. Alternatively, we provide computational results that suggest that while the greedy strategy proposed in \cite{bhawalkar2014online} is inoffensive for the worst case instance given $0 \leq z \leq 1$, it compromises on the competitive ratio obtained on the worst case instance given $z > 1$. Our experiments suggest that proposed algorithm improves on the competitive ratio obtained by the greedy strategy by some constant factor. 

The significance of this problem stems from the fact that, in some real-world applications, decisions are made after collecting a batch of data. Moreover, in many of these applications, batches of data are rarely produced by an absolute worst case adversary. The intent of our restricted adversarial model is to mimic the worst case instances that may effectively arise in the real world. For example, high-capacity ridesharing systems rely on the batch assignment of trip requests to vehicles \cite{alonso-Mora2017on-demand}, and some related problems such as that of optimal congregation points for passenger pickups and dropoffs can be modeled as a batched set cover problem. Intuitively, sequencing a batch of travel requests defeats the purpose of preparing the batch in the first place. Moreover, the batches that arise in this setting tend to have a VC-dimension greater than or equal to two, as the application revolves around exploiting the overlaps between distinct requests (see Section~\ref{sec: Preliminaries}). Our results formalize this intuition. As listed in \cite{bhawalkar2014online}, further examples of applications of the batched set cover problem may be found in distributed computing, facility planning, and subscription markets. 

To the best of our knowledge, we are the first to use the VC-dimension in the context of online (batched) covering problems. The VC-dimension has been used successfully in the context of approximation algorithms for (offline) set cover problems \cite{bronnimann1995almost,even2005hitting}, as well as in the context of improved running time bounds for unconstrained \cite{abraham2011vc} and constrained \cite{vera2017computing} shortest path algorithms. In both of these settings, the algorithms exploit the low VC-dimension of the set systems on which they operate. Perhaps surprisingly, algorithms for the batched set cover problem may instead exploit the high VC-dimension of the set systems on which they operate, which we model as a restriction on the adversary. Intuitively, the reason is that the adversary is forced to reveal complex, intertwined batches. A dedicated algorithm attempts to exploit the richness of the information revealed, while a greedy algorithm is myopic to the interactions between the set memberships of the elements in a batch.

\subsection{Organization}
In Section~\ref{sec: Preliminaries} we formally introduce our problems and definitions. In Section~\ref{sec: Fractional Online Set Cover} we consider bounds for the online fractional set cover problem. We present a known lower bound of $\Omega(H_m)$ on the competitive ratio of any online algorithm. While the tightness of this lower bound (up to constants) follows immediately from the existence of $O(H_m)$-competitive fractional algorithms \cite{buchbinder2005online,buchbinder2009design}, we present an inductive proof that shows the tightness of the lower bound without the explicit need of a competitive algorithm. This technique is used in Section~\ref{sec: Fractional Batched Set Cover} to show the tightness of a $\Omega(H_{m - 2^z + 1})$ lower bound on the competitive ratio of any batched algorithm, given our restricted adversary parametrized by $z$. The reason for doing this is our argument that batching may offer a constant factor improvement in the competitive ratio. Hence, tightness up to constants is not informative enough for our purposes. In Section~\ref{sec: Batched Algorithms} we formalize the greedy strategy suggested in \cite{bhawalkar2014online} and present our minor modification to an online algorithm found in \cite{buchbinder2009design}. We also present the results of our computational experiments.

\section{Preliminaries}
\label{sec: Preliminaries}

Let $x_j \in \{0,1\}$ be set to $1$ if $S_j$ is brought to the set cover and to $0$ otherwise. Now, consider \ref{lp: covering problem}, which describes the linear programming relaxation of the offline set cover problem. We refer to \ref{lp: covering problem} as the primal covering problem. Here, $c_j > 0$ refers to the cost of bringing some set $S_j \in \mathcal{S}$ to the set cover, and the objective is to minimize the total cost incurred. In the unweighted case, $c_j = 1$ for all $j = 1, \cdots, m$. Constraints $(1.1)$ ensure that every element $i = 1, \cdots, n$ in the ground set $X$ is covered. Note that the set membership information of each element is encoded in its respective constraint.
\begin{equation}
\tag{LP 1}
    \label{lp: covering problem}
    \begin{aligned}
    \begin{array}[t]{crl}
         \text{minimize} & \sum\limits_{j=1}^{m} c_j x_j &  \\
         \quad \text{s.t.} & & \\
         \quad (1.1) & \quad \sum\limits_{j: i \in S_j} x_j \geq 1, & i = 1, \cdots, n \\
         \quad     & \quad x_j \geq 0, & j = 1, \cdots, m \\
    \end{array}
    \end{aligned}
\end{equation}
The primal covering problem has an associated dual packing problem, described in \ref{lp: packing problem}. We refer to this primal-dual formulation throughout this work. We will refer to the collection of sets in $\mathcal{S}$ that individually contain $\sigma_i \in X$ by $\mathcal{S}(\sigma_i)$.
\begin{equation}
\tag{LP 2}
    \label{lp: packing problem}
    \begin{aligned}
    \begin{array}[t]{crl}
         \text{maximize} & \sum\limits_{i=1}^{n} y_i &  \\
         \quad \text{s.t.} & & \\
         \quad (2.1) & \quad \sum\limits_{i \in S_j} y_i \leq c_j, & j = 1, \cdots, m \\
         \quad     & \quad y_i \geq 0, & i = 1, \cdots, n \\
    \end{array}
    \end{aligned}
\end{equation}
In the fractional online setting \cite{buchbinder2005online,buchbinder2009design}, the objective function of \ref{lp: covering problem} is known a priori, but constraints $(1.1)$ are revealed one by one. This corresponds to the algorithm identifying the costs of the sets in $\mathcal{S}$ a priori, but the adversary revealing a sequence $\sigma = \sigma_1, \sigma_2, \cdots$ of elements of $X$, along with their respective set membership, in an online fashion. Equivalently, the right hand side of the constraints $(2.1)$ of \ref{lp: packing problem} are known a priori, but the variables involved in them and in the objective function are revealed one by one. 

Now, consider the following batched version of the set cover problem, which is also a game between an algorithm and an oblivious adversary. In the \textit{batched set cover problem}, $\mathcal{S}$ is identified a priori, but the adversary produces a sequence $\beta = \beta_1, \beta_2, \cdots$ of batches of elements of $X$, which it reveals one batch at a time. For instance, $\beta_k = \{\sigma_{k,1}, \cdots, \sigma_{k,|\beta_k|}\} \subseteq X$, where $|\beta_k|$ denotes the size of the $k$th batch. When a batch arrives, all of its elements, along with their respective set membership information, are revealed simultaneously. The fractional batched setting is analogous to the fractional online setting, except constraints $(1.1)$ appear in tandem. Equivalently, the variables involved in the objective function and constraints $(2.1)$ are revealed in tandem. Note that the online setting is trivially recovered when each batch is a singleton. We refer to the union of sets in $\mathcal{S}$ that individually cover the elements in $\beta_k = \{\sigma_{k,1}, \cdots, \sigma_{k,|\beta_k|}\}$, namely $\mathcal{S}(\sigma_{k,1}) \cup \cdots \cup \mathcal{S}(\sigma_{k,|\beta_k|})$, by $\mathcal{S}(\beta_k)$.

We define an instance $I$ of the online set cover problem as a collection $\mathcal{S}$ together with the adversarial sequence. We introduce the following performance measures. The batched setting for both of these measures is analogous.
\begin{definition}
An online algorithm $\text{ALG}^O$ is said to be $c$-competitive if for every instance $I$ of the problem it outputs a solution of cost at most $c \cdot \text{OPT}(I)$, where $\text{OPT}(I)$ is the cost of the optimal offline solution.
\end{definition}
\begin{definition}
An online adversary $\text{ADV}^O$ is said to be $c$-advantaged if it produces an instance $I$ such that every online algorithm $\text{ALG}^O$ outputs a solution of cost at least $c \cdot \text{OPT}(I)$, where $\text{OPT}(I)$ is the cost of the optimal offline solution.
\end{definition}

Our analysis in Section~\ref{sec: Fractional Batched Set Cover} relies on imposing a minimum on the VC-dimension of any batch $\beta_k$ produced by the adversary. The VC-dimension was first proposed by Vapnik and Chernovekis \cite{vapnik2015uniform}, and it is a widely used measure of complexity in computational learning theory. We work with the following definitions. 

\begin{definition}[Set System]
A set system $(X, \mathcal{S})$ is a ground set $X$ together with a collection $\mathcal{S}$ of subsets of $X$.
\end{definition}
\begin{definition}[Shattering]
A subset $B \subseteq X$ is said to be \textit{shattered} by $\mathcal{S}$ if $\{S \cap B : S \in \mathcal{S}\} = \mathcal{P}(B)$, where $\mathcal{P}(B)$ is the power set of $B$.
\end{definition}
\begin{definition}[VC-dimension]
The VC-dimension of a set system $(X, \mathcal{S})$ is the cardinality of the largest subset $B \subseteq X$ to be shattered by $\mathcal{S}$. We denote it by $\text{VCD}(X, \mathcal{S})$.
\end{definition}

In particular, upon the arrival of a batch $\beta_k$ we obtain a set system $(\beta_k, \mathcal{S})$, where $\mathcal{S}$ is known a priori. Moreover, note that \textit{i}) restricting the adversary to produce batches $\beta_k$ with VC-dimension $\text{VCD}(\beta_k, \mathcal{S}) \geq z $ is only meaningful when $m = |\mathcal{S}| \geq 2^{z}$; otherwise the adversary is unable to produce any batches, and \textit{ii}) by definition, any batch satisfying $\text{VCD}(\beta_k, \mathcal{S}) \geq z$ also satisfies $|\beta_k| \geq z$.  This is illustrated in Figure~\ref{fig: minimum m}, which showcases how a batch $\beta_k$ satisfying $\text{VCD}(\beta_k, \mathcal{S}) \geq z$ can be constructed with $m = 2^z$ and $|\beta_k| = z$. Observe that in each of the cases, $\beta_k$ is shattered since each of its subsets is the intersection of $\beta_k$ with some $S \in \mathcal{S}$. Of course, given $z$, there may be instances for which $m > 2^z$, or for which the adversary produces batches satisfying $|\beta_k| > z$, or both. We consider these cases in our analysis in Section~\ref{sec: Fractional Batched Set Cover}.

\begin{figure}
\begin{subfigure}{.45\textwidth}
\centering
\includegraphics[width=\linewidth]{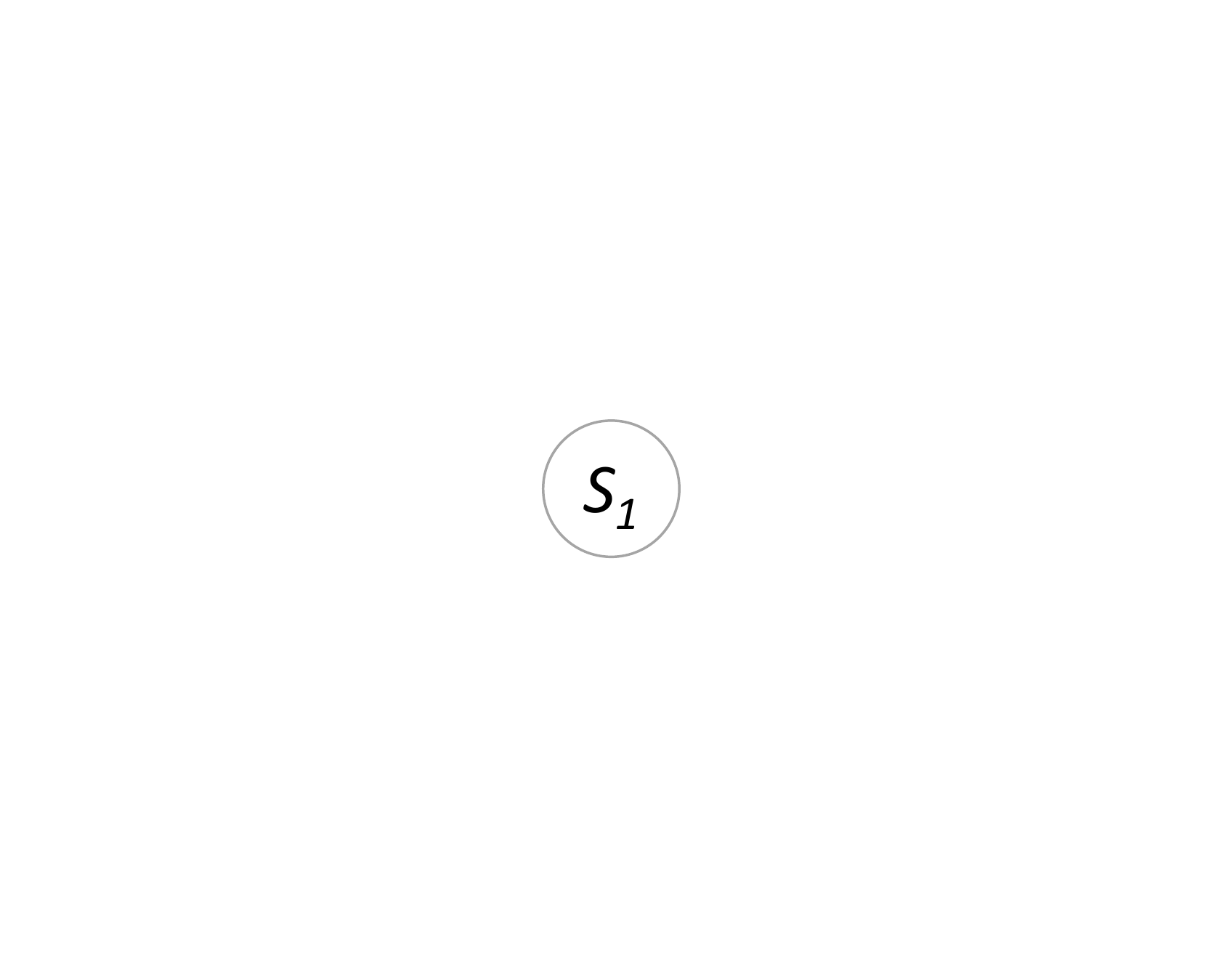}
\caption{$\text{VCD}(\beta_k, \mathcal{S}) \geq 0$.}
\end{subfigure}\hfill
\begin{subfigure}{.45\textwidth}
\centering
\includegraphics[width=\linewidth]{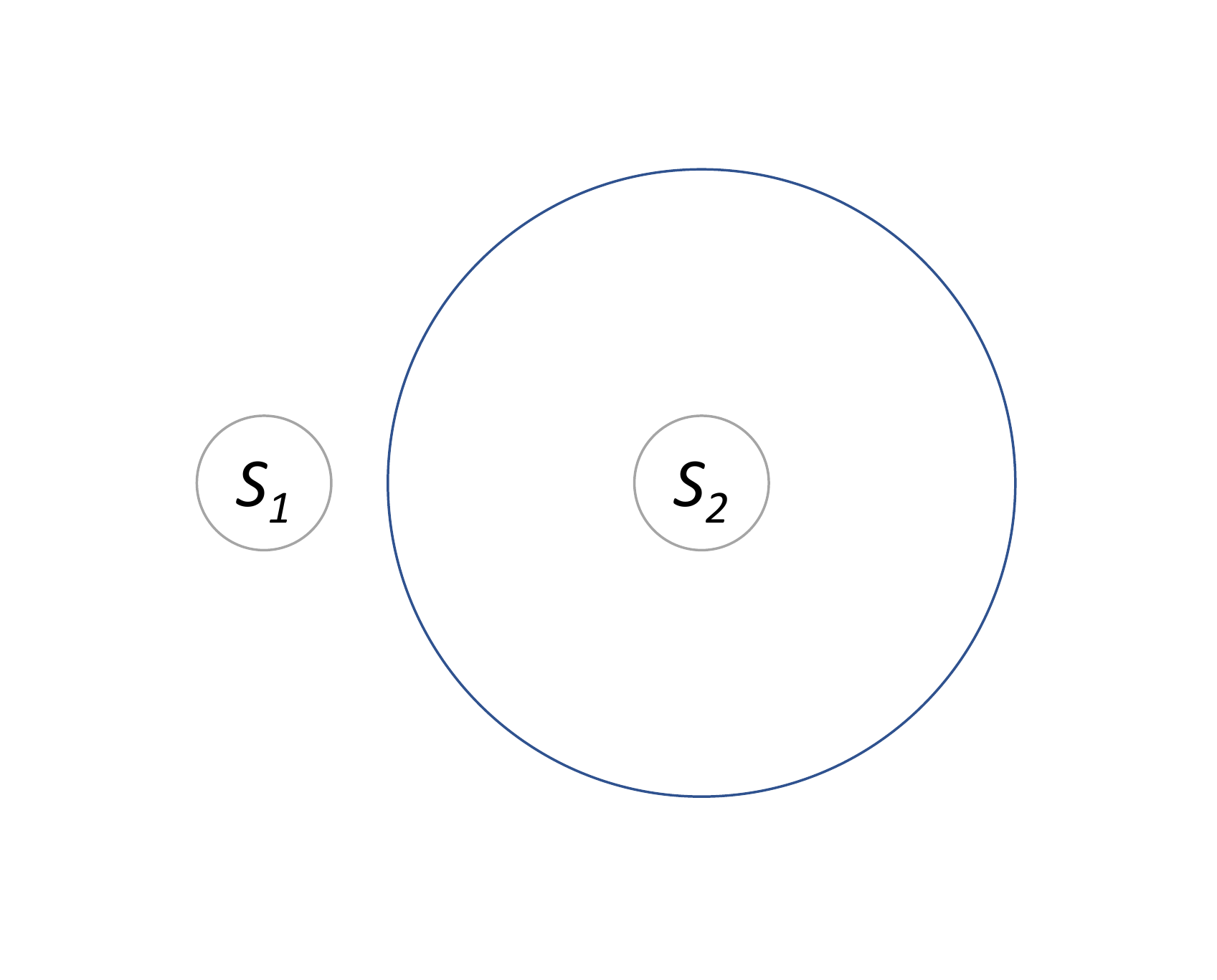}
\caption{$\text{VCD}(\beta_k, \mathcal{S}) \geq 1$.}
\end{subfigure}\hfill
\begin{subfigure}{.45\textwidth}
\centering
\includegraphics[width=\linewidth]{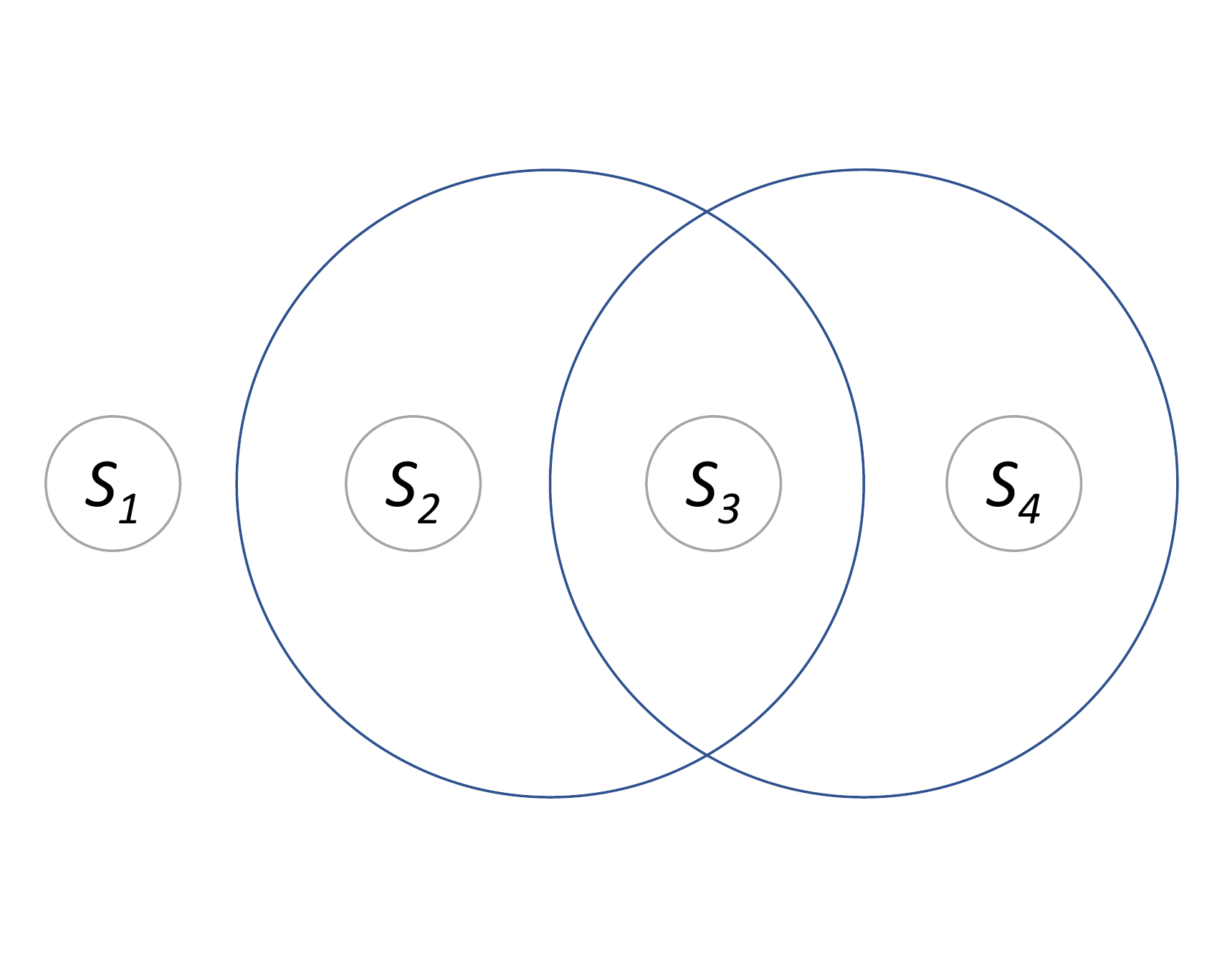}
\caption{$\text{VCD}(\beta_k, \mathcal{S}) \geq 2$.}
\end{subfigure}\hfill
\begin{subfigure}{.45\textwidth}
\centering
\includegraphics[width=\linewidth]{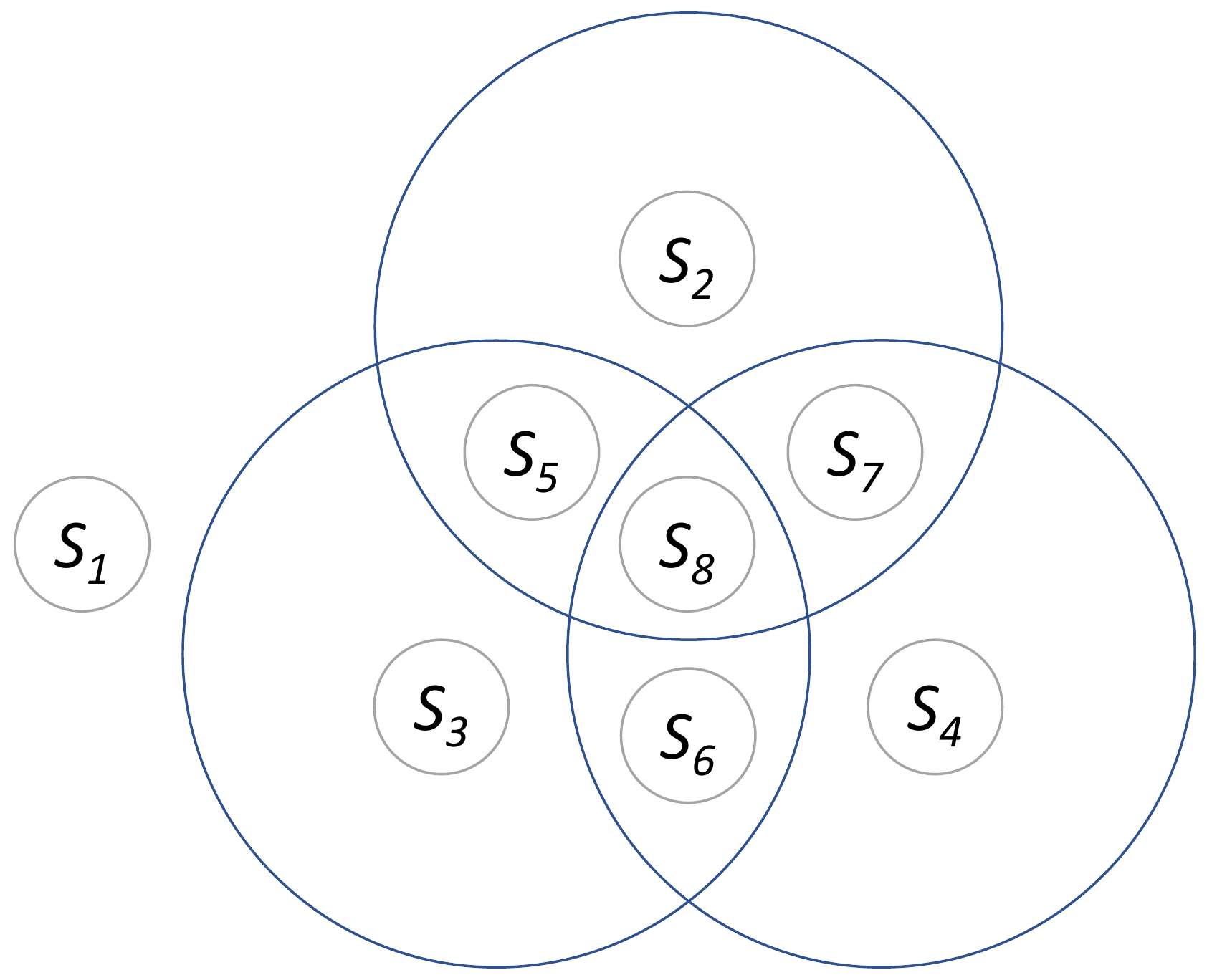}
\caption{$\text{VCD}(\beta_k, \mathcal{S}) \geq 3$.}
\end{subfigure}

\caption{Construction of $\mathcal{S}$ and $\beta_k$ satisfying $\text{VCD}(\beta_k, \mathcal{S}) \geq z$, for $z = 0, 1, 2, 3$, with the minimum possible cardinality requirements on $m = |\mathcal{S}|$ and $|\beta_k|$. The construction for $z > 3$ is analogous. The nodes labeled $S_1, \cdots, S_{m}$ represent the sets in $\mathcal{S}$, whereas the circles around them represent the elements of $\beta_k \subseteq X$ they contain (i.e., constraints in \ref{lp: covering problem}).}
\label{fig: minimum m}
\end{figure}

In a ridesharing context, we may interpret each $S \in \mathcal{S}$ as a possible congregation point (e.g., an intersection in a street network), whereas each constraint corresponds to the set of compatible congregation points (e.g., within walking distance) for each trip origin or destination. Note that the only reasons why we would have $\text{VCD}(\beta_k, \mathcal{S}) \leq 1$ are if \textit{i}) $|\beta_k| \leq 1$, or \textit{ii}) $|\beta_k| > 1$ but all the travel request form either a collection of proper subsets or a collection completely disjoint subsets of the possible congregation points. Given sufficiently high and heterogeneous demand, batches with $\text{VCD}(\beta_k, \mathcal{S}) \leq 1$ are unlikely to arise; the batches that arise look more like the those in Figure~\ref{fig: minimum m} (c) and (d). 


\section{Fractional Online Set Cover}
\label{sec: Fractional Online Set Cover}

\begin{lemma}[Variation of Buchbinder and Naor \cite{buchbinder2009design}]
\label{lemma: online lower bound}
There exists an instance $I^*$ of the unweighted fractional online set cover problem such that any online algorithm is $\Omega\left(H_m\right)$-competitive on this instance.
\end{lemma}
\begin{proof}
Consider the following instance $I^*$, which is particular to the sequence $\sigma$ produced by an adversary in response to an arbitrary online algorithm $\text{ALG}^O$. Let $\sigma_1 \in \bigcap_{S_j \in \mathcal{S}} S_j$. Upon the arrival of $\sigma_1$, $\text{ALG}^O$ must satisfy $\sum_{j: \sigma_1 \in S_j} x_j \geq 1$. Thus, it must let $x_j \geq 1/m$ for at least one $S_j \in \mathcal{S}$. Refer to such $S_j \in \mathcal{S}$ as $S^1$ and to its corresponding variable as $x^1$. Now, let $\sigma_2 \in \bigcap_{S_j \in \mathcal{S} \setminus S^1} S_j$. Upon the arrival of $\sigma_2$, $\text{ALG}^O$ must satisfy $\sum_{j: \sigma_2 \in S_j} x_j \geq 1$. Thus, it must let $x_j \geq 1/(m-1)$ for at least one $S_j \in \mathcal{S} \setminus S^1$. Again, refer to such $S_j \in \mathcal{S}\setminus S^1$ as $S^2$ and to its corresponding variable as $x^2$. In general, an adversary may continue revealing elements $\sigma_i$ satisfying
\begin{align*}
    \sigma_i \in \bigcap_{S_j \in \mathcal{S} \setminus \bigcup_{i' < i} S^{i'}} S_j,
\end{align*}
forcing the algorithm to let $x_j \geq 1/(m - i + 1)$ for at least one $S_j \in \mathcal{S} \setminus \bigcup_{i' < i} S^{i'}$. Refer to such $S_j \in \mathcal{S} \setminus \bigcup_{i' < i} S^{i'}$ as $S^i$ and to its corresponding variable as $x^i$. After $m$ steps, the total cost $\text{ALG}^O(I^*)$ incurred by the algorithm, namely $x^1 + \cdots + x^m$, is at least
\begin{align*}
    \frac{1}{m} + \frac{1}{m - 1} + \cdots + \frac{1}{2} + 1 = H_m. 
\end{align*}
Meanwhile, the total cost $OPT(I^*)$ incurred by an optimal offline solution is $1$, which corresponds to simply letting $x^m = 1$. Thus, $H_m \leq \frac{\text{ALG}^O(I^*)}{OPT(I^*)}$.
\end{proof}

Note that $I^*$ depends on $\text{ALG}^O$ only in the sense that the particular adversarial sequence $\sigma$ produced is a response to the particular algorithm; the lower bound on the competitive factor, on the other hand, is independent of the algorithm. Thus, we may parametrize the instance $I^*$ in Lemma~\ref{lemma: online lower bound} by $m = |\mathcal{S}|$ to obtain the instance class $I^*(m)$. In other words, $I^*(m)$ refers to the instances that produce a lower bound of $H_m$ on the competitive factor of any online algorithm as a result of the adversary following the strategy in the proof of Lemma~\ref{lemma: online lower bound}. If we vary $m$, we obtain the family of instance classes $\mathcal{I}^* = \{I^*(m) : m \in \mathbb{Z}_+\}$.

The tightness of this lower bound (up to constants) is an immediate result of the existence of $O(\log m)$-competitive algorithms for the fractional online set cover problem \cite{buchbinder2005online,buchbinder2009design}. In Lemma~\ref{lemma: online upper bound} we present a different approach to show that this lower bound is tight, this time without relying on a particular algorithm. We use Lemma~\ref{lemma: H_m}, whose proof is in the Appendix.

\begin{lemma}
\label{lemma: H_m}
$H_r > \frac{1}{2} \left( H_{r-t} + H_t \right)$ for any integers $r \geq t \geq 0$.
\end{lemma}

\begin{lemma}
\label{lemma: online upper bound}
There exists an online algorithm $\text{ALG}^O$ for the unweighted fractional online set cover problem such that any adversary is $O(H_m)$-advantaged. In particular, this bound is matched by any adversary that follows the strategy in the family of instance classes $\mathcal{I}^*$, described in the proof of Lemma~\ref{lemma: online lower bound}.
\end{lemma}
\begin{proof}
By Lemma~\ref{lemma: online lower bound}, there exists an adversary that is $H_m$-advantaged, namely one that follows the strategy of instance class $I^*(m)$. We need to show that this is indeed the best an adversary can be guaranteed to achieve. We prove this by strong induction on $m$ and by focusing on an arbitrary adversary $\text{ADV}^O$. We will make use of the existence of a randomized algorithm $\text{ALG}^O$ that, in principle, produces specific outcomes with non-zero probability.

\textit{Base Case ($m=1$)}: When $\text{ADV}^O$ reveals any element $\sigma_1$, it must be the case that $\sigma_1 \in S_1$. Then, $\text{ALG}^O$ must let $x_1 = 1$, achieving a competitive factor of $H_1$, as in $I^*(1)$
    
\textit{Inductive Step}: Assume, by way of strong induction, than the statement is true for $m = 2, \cdots, w$. We need to show that the statement is true for $m = w + 1$. By Lemma~\ref{lemma: online lower bound}, there exists an adversary that is $H_{w+1}$-advantaged, namely one that follows the strategy of instance class $I^*(w+1)$. Now, consider the case in which $\text{ADV}^O$ deviates from the strategy of $I^*(w+1)$ on some arbitrary step $i$. Let $\sigma_{i^*}$ be the $i$th element according to the strategy of $I^*(w+1)$. If  $\text{ADV}^O$ reveals an element $\sigma_i$ such that $\mathcal{S}(\sigma_i) \cap \mathcal{S}(\sigma_{i^*})$ is empty, the cost of the optimal solution increases by 1, which by Lemma~\ref{lemma: H_m} irrevocably decreases the advantage of $\text{ADV}^O$. Therefore, suppose that $\text{ADV}^O$ reveals an element $\sigma_i$ such that $\mathcal{S}(\sigma_i) \cap \mathcal{S}(\sigma_{i^*})$ is non-empty. Then, with non-zero probability $\text{ALG}^O$ disregards all $S_j \in \mathcal{S}(\sigma_i) \setminus \mathcal{S}(\sigma_{i^*})$, if any, making such deviation futile. Thus, safely assume that $\text{ADV}^O$ instead reveals an element $\sigma_i$ such that $\mathcal{S}(\sigma_i) \subset \mathcal{S}(\sigma_{i^*})$. Let $t = |\mathcal{S}(\sigma_i)|$, $r = |\mathcal{S}(\sigma_{i^*})|$, and note that $r > t$. Let $i' > i$ be the first step after step $i$ such that $\mathcal{S}(\sigma_{i}) \cap \mathcal{S}(\sigma_{i'}) \neq \emptyset$. As before, with non-zero probability $\text{ALG}^O$ disregards all $S_j \in \mathcal{S}(\sigma_{i'}) \setminus \mathcal{S}(\sigma_i)$, if any, so safely assume that $\text{ADV}^O$ reveals an element $\sigma_{i'}$ such that $\mathcal{S}(\sigma_{i'}) \subset \mathcal{S}(\sigma_{i})$. Then, by the inductive hypothesis, given that element $\sigma_i$ satisfied  $\mathcal{S}(\sigma_i) \subset \mathcal{S}(\sigma_{i^*})$, the best $\text{ADV}^O$ can do is to recreate $I^*(t)$ on the remainder of the steps, starting with $i'$. In particular, the best $\text{ADV}^O$ can do is to reveal an element $\sigma_{i'}$ such that $\mathcal{S}(\sigma_{i'}) \subset \mathcal{S}(\sigma_{i})$ and $|\mathcal{S}(\sigma_{i'})| + 1 = |\mathcal{S}(\sigma_{i})|$. A symmetric argument can be made about concurrently recreating $I^*(r-t)$ on the remainder of the steps, which is disjoint from $I^*(t)$ after the $i$th step and hence increases the offline solution by one. However, by Lemma~\ref{lemma: H_m}, this achieves a strictly lower competitive advantage for $\text{ADV}^O$.

\end{proof}

\section{Fractional Batched Set Cover}
\label{sec: Fractional Batched Set Cover} 

\subsection{General Case}
\begin{lemma}
\label{lemma: batched lower bound}
There exists an instance $I^*$ of the unweighted fractional batched set cover problem such that any batched algorithm is $\Omega\left(H_m\right)$-competitive on this instance.
\end{lemma}
\begin{lemma}
\label{lemma: batched upper bound}
There exists a batched algorithm $\text{ALG}^B$ for the unweighted fractional batched set cover problem such that any adversary is $O(H_m)$-advantaged.
\end{lemma}
Lemma~\ref{lemma: batched lower bound} follows from the fact that the fractional online set cover problem is a special case of the fractional batched set cover problem, together with Lemma~\ref{lemma: online lower bound}. In Section \ref{sec: Restricted Adversary} we consider the case in which the adversary is imposed a minimum VC-dimension for any batch $\beta_k$ produced. In the proof of Lemma~\ref{lemma: batched restricted upper bound}, we mention why an adversary never benefits from producing batches with a VC-dimension larger than the minimum required. This, together with Lemma~\ref{lemma: online upper bound}, yields Lemma~\ref{lemma: batched upper bound}.

\subsection{Restricted Adversary}
\label{sec: Restricted Adversary}

Our intent now is to characterize instance classes that distinguish the fractional batched set cover problem from the fractional online set cover problem. In particular, we restrict the adversary in that it is forced to produce batches $\beta_k$ satisfying $\text{VCD}(\beta_k, \mathcal{S}) \geq z $, for some $z \in \mathbb{N}^0$. Given $z$, we assume that $m = |\mathcal{S}| \geq 2^z$ and $|\beta_k| \geq z$, as described in Section~\ref{sec: Preliminaries}. 
\begin{lemma}
\label{lemma: batched restricted lower bound}
There exists an instance $I_z^*$ of the unweighted fractional batched set cover problem, with an adversary satisfying $\text{VCD}(\beta_k, \mathcal{S}) \geq z $ for any batch $\beta_k$, such that any batched algorithm is $\Omega\left(H_{m - 2^z + 1}\right)$-competitive on this instance.
\end{lemma}
\begin{proof}
This proof is analogous to that of Lemma~\ref{lemma: online lower bound}. Consider the following instance $I_z^*$, which is particular to the sequence $\beta$ produced by an adversary satisfying $\text{VCD}(\beta_k, \mathcal{S}) \geq z $ for any batch $\beta_k$, in response to an arbitrary batched algorithm $\text{ALG}^B$. In the following, our ordering of $S_1, \cdots, S_m \in \mathcal{S}$ is arbitrary.

Let $\beta_1 = \{\sigma_{1,1}, \cdots, \sigma_{1,z}, \sigma_{1,z+1}\}$ such that $\mathcal{S}(\{\sigma_{1,1}, \cdots, \sigma_{1,z}\})$ is as in the diagrams in Figure~\ref{fig: minimum m} on sets $S_1, \cdots, S_{2^z- 1},S_{2^z}$ with $\beta_1 \cap S_{2^z} = \beta_1$, in addition to each $\sigma_{1,q} \in \beta_1$ being also contained in each $S_j \in \{S_{2^z}, S_{2^z + 1}, \cdots, S_{m}\}$. This last property cannot decrease the VC-dimension, so $\beta_1$ is a valid batch. Then, because of the constraint corresponding to $\sigma_{1,z+1} \in \beta_1$, $\text{ALG}^B$ must let $x_j \geq \frac{1}{m - 2^z + 1}$ for at least one $S_j \in \{S_{2^z}, \cdots, S_{m}\}$. For clarity and without loss of generality, assume such $S_j$ is $S_{2^z}$.

Then, let $\beta_2 = \{\sigma_{2,1}, \cdots, \sigma_{2,z}, \sigma_{2,z+1}\}$ such that $\mathcal{S}(\{\sigma_{2,1}, \cdots, \sigma_{2,z}\})$ is as in the diagrams in Figure~\ref{fig: minimum m} on sets $S_2, \cdots, S_{2^z}, S_{2^{z}+1}$ with $\beta_2 \cap S_{2^{z}+1} = \beta_2$, in addition to each $\sigma_{2,q} \in \beta_2$ being also contained in each $S_j \in \{S_{2^z+1}, \cdots, S_{m}\}$. Then, because of the constraint corresponding to $\sigma_{2,z+1} \in \beta_2$, $\text{ALG}^B$ must let $x_j \geq \frac{1}{m - 2^z - 1 + 1}$ for at least one $S_j \in \{S_{2^z+1}, \cdots, S_{m}\}$. For clarity and without loss of generality, assume such $S_j$ is $S_{2^{z}+1}$.

In general, an adversary may continue revealing $\beta_k = \{\sigma_{k,1}, \cdots, \sigma_{k,z}, \sigma_{k,z+1}\}$ such that $\mathcal{S}(\{\sigma_{k,1}, \cdots, \sigma_{k,z}\})$ is as in the diagrams in Figure~\ref{fig: minimum m} on sets $S_{k}, \cdots$, $S_{2^z + k - 2}, S_{2^z + k - 1}$ with $\beta_k \cap S_{2^z + k - 1} = \beta_k$, in addition to each $\sigma_{k,q} \in \beta_k$ being also contained in each $S_j \in \{S_{2^z + k - 1}, \cdots, S_{m}\}$. Then, because of the constraint corresponding to $\sigma_{k,z+1} \in \beta_k$, $\text{ALG}^B$ must let $x_j \geq \frac{1}{m - 2^z - (k - 1) + 1}$ for at least one $S_j \in \{S_{2^z + k - 1}, \cdots, S_{m}\}$. For clarity and without loss of generality, assume such $S_j$ is $S_{2^z + k - 1}$.

After the $m - 2^z + 1$ possible steps, the total cost $\text{ALG}^B(I_z^*)$ incurred by the algorithm, namely $x_{2^z} + \cdots + x_{m}$, is at least
\begin{align*}
    H_{m - 2^z + 1} = \frac{1}{m - 2^z + 1} + \frac{1}{m - 2^z} + \cdots + \frac{1}{2} + 1.
\end{align*}
Meanwhile, the total cost $OPT(I_z^*)$ incurred by an optimal offline solution is $1$, which corresponds to simply letting $x^{m - 2^z + 1} = 1$. Thus, $H_{m - 2^z + 1} \leq \frac{\text{ALG}^O(I_z^*)}{OPT(I_z^*)}$. In simple terms, the adversary may capitalize on $S_j \in \{S_{2^z}, \cdots, S_{m}\}$ while assuming a sunk cost on $S_j \in \{S_1, \cdots, S_{2^z - 1}\}$.

\end{proof}
As in Section~\ref{sec: Fractional Online Set Cover}, we parametrize the instances $I_z^*$ in Lemma~\ref{lemma: batched restricted lower bound} by $m$. Thus, given $z$, we obtain the family of instance classes $\mathcal{I}_z^* = \{I_z^*(m) : m \in \mathbb{Z}_+\}$. We recover the general case when $z = 0$. Analogous to Lemma~\ref{lemma: online upper bound}, Lemma~\ref{lemma: batched restricted upper bound} shows the lower bound given $z$ is tight. Its proof can be found in the Appendix.


\begin{lemma}
\label{lemma: batched restricted upper bound}
There exists a batched algorithm $\text{ALG}^B$ for the unweighted fractional online set cover problem such that any adversary satisfying $\text{VCD}(\beta_k, \mathcal{S}) \geq z $ for any batch $\beta_k$ is $O\left(H_{m - 2^z + 1}\right)$-advantaged. In particular, this bound is matched by any adversary that follows the strategy in the family of instance classes $\mathcal{I}_z^*$, described in the proof of Lemma~\ref{lemma: batched restricted lower bound}.
\end{lemma}

\section{Batched Algorithms}
\label{sec: Batched Algorithms}
\subsection{Analysis}
\label{sec: Analysis}

Note that since any batch $\beta_k$ could be artificially decomposed into a sequence of $|\beta_k|$ elements that are `revealed' one at a time, any competitive algorithm for online set cover would produce a competitive feasible solution; we refer to this as the `trivial' approach. More precisely, the trivial approach consists of two phases: \textit{i}) some (possibly randomized) subroutine that executes a mapping $f : \beta_k \to ( \sigma_{k,1'}, \cdots, \sigma_{k,b_k'})$, where $\sigma_{k,q'}$ is the $q'$th element of the artificial sequence, followed by \textit{ii}) any competitive algorithm for the online set cover problem. 

An example of such approach is the $O(\log m)$-competitive primal-dual subroutine in Algorithm~\ref{algo: Trivial Approach}, which we refer to as $\text{ALG}^{B,T}$, followed by any $O(\log n)$ rounding technique (i.e., the second phase of the primal-dual method). $\text{ALG}^{B,T}$ is a minor modification of the $O(\log m)$-competitive Algorithm 2 in Section 4.2 of \cite{buchbinder2009design} for the batched setting. Note that $d = \max_{k,q} |\mathcal{S}(\sigma_{k,q})| \leq m$. Its correctness follows immediately from Theorem 4.2 of \cite{buchbinder2009design}. For conciseness, we only mention that the proof relies on showing three claims together with weak duality: \textit{i}) the algorithm produces a primal feasible solution to \ref{lp: covering problem}, \textit{ii}) the algorithm produces a dual feasible solution to \ref{lp: packing problem}, and \textit{iii}) the objective value of \ref{lp: covering problem} is bounded above by $O(\log m)$ times the objective value of \ref{lp: packing problem}. Clearly, the complete algorithm is $O(\log m \log n)$-competitive.

\begin{algorithm}[ht]
\SetAlgoLined
\tcc{Upon the arrival of batch $k := \{ \sigma_{k,1}, \cdots, \sigma_{k,|\beta_k|}\}$:}
\For{$\sigma_{k,q} \in f\left( \{ \sigma_{k,1}, \cdots, |\beta_k| \} \right)$}{
    \While{$\sum\limits_{j: \sigma_{k,q} \in S_j} x_j < 1$}{
        \text{Increase $y_{k,q}$ continuously} \\
        $x_j \gets \frac{1}{d} \left[ \exp\left(\frac{\ln{\left(1+d\right)}}{c_j} \sum\limits_{\substack{k', q' : \\ \sigma_{k',q'} \in S_j}} y_{k', q'} \right) - 1 \right], \forall j: \sigma_{k,q} \in S_j$ \\
    }
}
\caption{$\text{ALG}^{B,T}$: `Trivial' batched algorithm.}
\label{algo: Trivial Approach}
\end{algorithm}

Nevertheless, unless $|\beta_k| = 1$ for all $k$, the trivial approach may fail to leverage the rich information that is possibly implicit in the fact that all elements of a given batch are revealed simultaneously. We refer to an algorithm that attempts to leverage any such information as a `dedicated' algorithm. We obtain such an algorithm if we replace $\text{ALG}^{B,T}$ with Algorithm~\ref{algo: Batched Approach}, which we refer to as $\text{ALG}^{B,D}$. 

\begin{algorithm}[ht]
\SetAlgoLined
\tcc{Upon the arrival of batch $k := \{\sigma_{k,1}, \cdots, \sigma_{k,|\beta_k|}\}$:}
\While{$\exists \sigma_{k,q} \text{ such that } \sum\limits_{j: \sigma_{k,q} \in S_j} x_j < 1$}{
    \text{Increase $y_{k,q}$ continuously}, $\forall q: \sum\limits_{j: \sigma_{k,q} \in S_j} x_j < 1$ \\
    $x_j \gets \frac{1}{d} \left[ \exp\left(\frac{\ln{\left(1+d\right)}}{c_j} \sum\limits_{\substack{k', q' : \\ \sigma_{k',q'} \in S_j}} y_{k', q'} \right) - 1 \right], \forall j: \exists \sigma_{k,q} \in S_j$ \\
}
\caption{$\text{ALG}^{B,D}$: `Dedicated' batched algorithm.}
\label{algo: Batched Approach}
\end{algorithm}

Note the difference in how the dual variables are updated: sequentially in $\text{ALG}^{B,T}$ and simultaneously in $\text{ALG}^{B,D}$. This is reminiscent of the approach of increasing multiple variables at once in a primal-dual algorithm by \cite{williamson1995primaldual} for the generalized Steiner tree problem. As expected, $\text{ALG}^{B,D}$ is also $O(\log m)$-competitive\footnote{To be precise, both algorithms are $2 \ln(1 + d)$-competitive.}; it is also a minor modification of the $O(\log m)$-competitive Algorithm 2 in Section 4.2 of \cite{buchbinder2009design}, and its correctness also follows immediately from Theorem 4.2 of \cite{buchbinder2009design}. Unfortunately, we are unable to improve its analysis in a way that reflects the intuition that batching should improve the competitive factor obtained under certain conditions. For example, we expect batching to help in the families of instances $\mathcal{I}_z^*$, shown in Section~\ref{sec: Fractional Batched Set Cover} to produce a more generous lower bound on the competitive factor of any algorithm due to the VC-dimension requirement on the batches produced by the adversary. We leave presenting such an analysis as an open problem. As an alternative, in Section~\ref{sec: Computational Experiments} we present the results of computational experiments that compare the performance (i.e., competitive factor) of $\text{ALG}^{B,T}$ and $\text{ALG}^{B,D}$ on instances of $\mathcal{I}_z^*$.

\subsection{Computational Experiments}
\label{sec: Computational Experiments}

Figure~\ref{fig: experiments} (a) and (b) present the results of computational experiments that compare the worst-case performance (i.e., competitive factor) of $\text{ALG}^{B,T}$ and $\text{ALG}^{B,D}$ obtained on instances of $\mathcal{I}_z^*$, respectively, for various values of $z$ and $m$. We discretize both algorithms with a step size of $\epsilon = 0.001$. We justify the use of instances of $\mathcal{I}_z^*$ by the fact that it is not some arbitrary family of instance classes; Lemma~\ref{lemma: batched restricted lower bound} ans Lemma~\ref{lemma: batched restricted upper bound} imply it produces tight bounds. Also, note that because of the symmetric nature of the batches $\beta_k$ that arise in $\mathcal{I}_z^*$, the particular order produced by the mapping $f : \beta_k \to ( \sigma_{k,1'}, \cdots, \sigma_{k,|\beta_{k}|'})$ in $\text{ALG}^{B,T}$ is irrelevant for evaluating the worst-case performance of the trivial approach so long as $\sigma_{k,z+1}$ is the last element `revealed'. This notion does not translate to $\text{ALG}^{B,D}$, since the dual variable updates occur simultaneously. Also, recall that for instances in $\mathcal{I}_z$ the optimal offline solution is one.

\begin{figure}
\centering
\begin{subfigure}{0.5\textwidth}
\centering
\includegraphics[width=\linewidth]{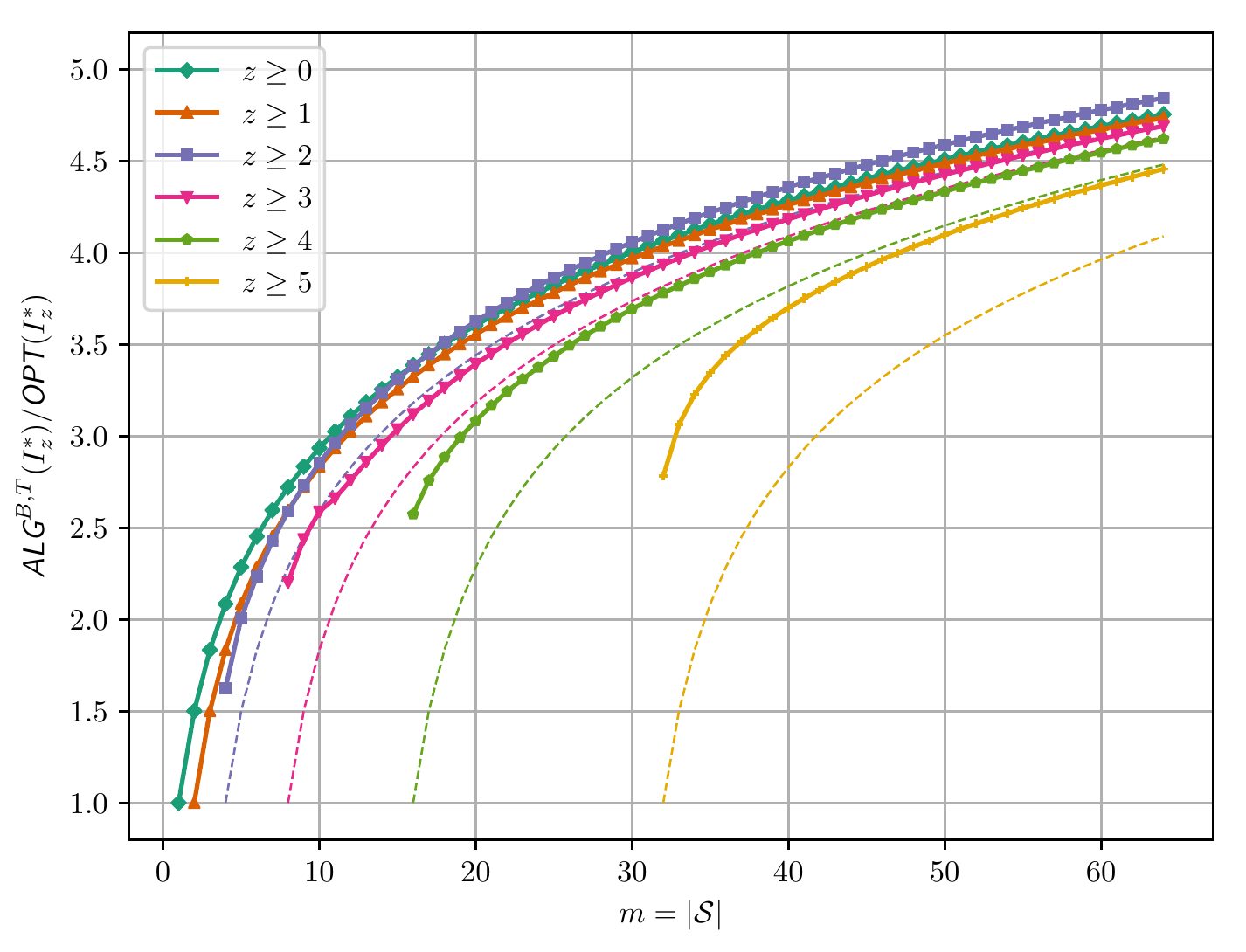}
\caption{$\text{ALG}^{B,T}$.}
\end{subfigure}\hfill
\begin{subfigure}{0.5\textwidth}
\centering
\includegraphics[width=\linewidth]{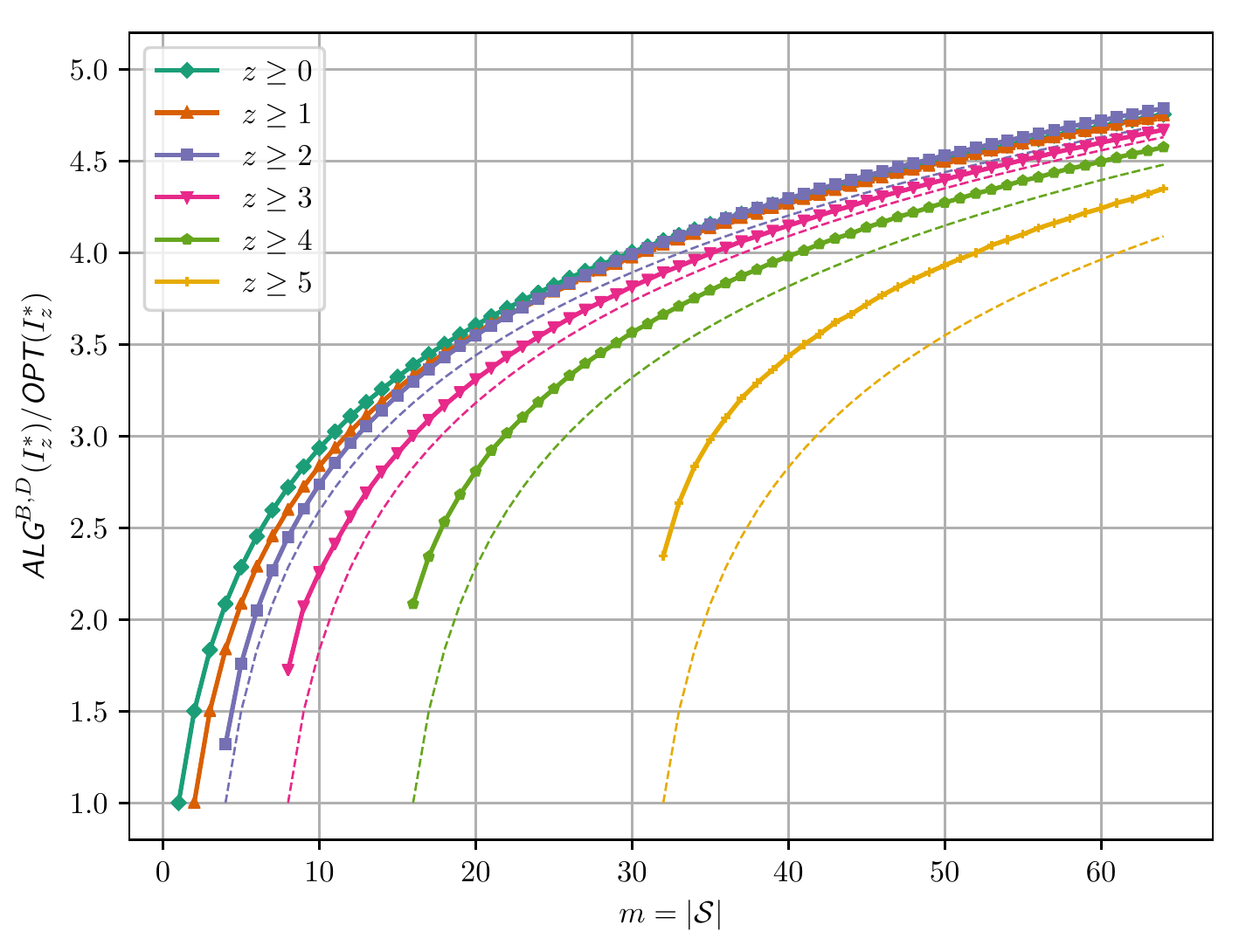}
\caption{$\text{ALG}^{B,D}$.}
\end{subfigure}\hfill
\caption{Competitive ratio of (a) $\text{ALG}^{B,T}$ and (b) $\text{ALG}^{B,D}$ obtained on instances of $\mathcal{I}_z^*$ for $z = 0, 1, 2, 3, 4$. The solid line indicates to the competitive ratio obtained by the algorithm, while the dashed line indicates the $\Omega(H_{m - 2^z + 1})$ lower bound on the competitive ratio of any batched algorithm.}
\label{fig: experiments}
\end{figure}

Note that for $z = 0, 1$, both algorithms match the theoretical lower bound. This is expected, as in both of these cases, per the description of $\mathcal{I}_z^*$, all the batches are singletons. On the other hand, for $z > 1$, neither of the algorithms match said lower bound. However, it can be observed that $\text{ALG}^{B,D}$ is closer than $\text{ALG}^{B,T}$ to the theoretical lower bound on the competitive ratio; this is not surprising as the algorithm has only be shown to be optimal up to constants. The improvement is more significant as $z$ increases, though it is only by a constant factor; both curves remain logarithmic. These results provide empirical evidence of the benefits of batching when the VC-dimension is high.

\section*{Acknowledgements}
The first author was funded through the Center for Transportation, Environment, and Community Health (CTECH) as well as a Federal Highway Administration (FHWA) Dwight David Eisenhower Transportation Fellowship.

\newpage

%
%
\bibliographystyle{splncs04}
\bibliography{bib}

\newpage

\section*{Appendix}
\subsection{Proof of Lemma~\ref{lemma: H_m}.}
\begin{proof}
Assume, without loss of generality, that $r - t \geq t$. Note that
\begin{align*}
    H_r 
    &= \frac{1}{r} + \frac{1}{r-1} + \cdots + \frac{1}{t+1} + H_t.
\end{align*}
Likewise, note that
\begin{align*}
    \frac{1}{2} \left( H_{r-t} + H_t \right) 
    &= \frac{1}{2} \left(\frac{1}{r-t} + \frac{1}{r-t - 1} + \cdots + \frac{1}{2} + 1 + 
    \frac{1}{t} + \frac{1}{t-1} + \cdots + \frac{1}{2} + 1\right) \\
    &= \frac{1}{2} \left(\frac{1}{r-t} + \frac{1}{r-t-1} + \cdots \frac{1}{t+1}\right) + \frac{1}{t} + \frac{1}{t-1} + \cdots + \frac{1}{2} + 1 \\
    &= \frac{1}{2} \left(\frac{1}{r-t} + \frac{1}{r-t-1} + \cdots \frac{1}{t+1}\right) + H_t.
\end{align*}
Clearly,
\begin{align*}
    \frac{1}{r} + \frac{1}{r-1} + \cdots + \frac{1}{t+1} > \frac{1}{2} \left(\frac{1}{r-t} + \frac{1}{r-t-1} + \cdots \frac{1}{t+1}\right),
\end{align*}
so the proof is complete.
\end{proof}

\subsection{Proof of Lemma~\ref{lemma: batched restricted upper bound}.}

\begin{proof}
This proof is analogous to that of Lemma~\ref{lemma: online upper bound}. By Lemma~\ref{lemma: batched restricted upper bound}, there exists an adversary satisfying $\text{VCD}(\beta_k, \mathcal{S}) \geq z $ for any batch $\beta_k$ that is $H_{m - 2^z + 1}$-advantaged, namely one that follows the strategy of instance class $I_z^*(m)$. We need to show that this is indeed the best an adversary can be guaranteed to achieve. We prove this by strong induction on $m$ and by focusing on an arbitrary batched adversary $\text{ADV}^B$. We will make use of the existence of a randomized algorithm $\text{ALG}^B$ that, in principle, produces specific outcomes with non-zero probability.

\textit{Base Case ($1 \leq m < 2^z$)}. If $m < 2^z$, the statement is vacously true because the adversary cannot produce any batches.

\textit{Base Case ($m=2^z$)}. When $\text{ADV}^B$ reveals any batch $\beta_1 = \{\sigma_{1,1}, \cdots, \sigma_{1,z}\}$, it must be the case that $\mathcal{S}(\beta_1)$ is as in the diagrams in Figure~\ref{fig: minimum m} on sets $S_1, \cdots, S_{m=2^z}$. Then, with non-zero probability, $\text{ALG}^B$ disregards any $S_j \notin \{S_{m}\}$. In such case, $\text{ALG}^B$ lets $x_{m} = 1$, achieving a competitive factor of $H_1$. This is in agreement with the description of $I_z^*(1)$.
    
\textit{Inductive Step}. Assume, by way of strong induction, than the statement is true for $m = 2^z + 1, \cdots, w$. We need to show that the statement is true for $m = w + 1$. By Lemma~\ref{lemma: batched restricted lower bound}, there exists an adversary that is $H_{w+1 - 2^z + 1}$-advantaged, namely one that follows the strategy of instance class $I_z^*(w+1)$. Now, consider the case in which $\text{ADV}^B$ deviates from the strategy of $I_z^*(w+1)$ on some arbitrary step $k$. Let $\beta_k= \{\sigma_{k,1}, \cdots, \sigma_{k,l}\}$, where $l \geq z$, and let $\beta_{k^*}= \{\sigma_{k^*,1}, \cdots, \sigma_{k^*,z}, \sigma_{k^*,z + 1}\}$ be the $k$th batch according to the strategy of $I_z^*(w+1)$. Further, denote  $\mathcal{S}^{\cap \beta_k} = \cap_{q = 1}^l \mathcal{S}(\sigma_{k,q})$, with $t=|\mathcal{S}^{\cap \beta_k}|$, as well as $\mathcal{S}^{\cap \beta_{k*}} = \cap_{q = 1}^z \mathcal{S}(\sigma_{k^*,z})$, with $r=|\mathcal{S}^{\cap \beta_{k*}}|$.  If $\mathcal{S}^{\cap \beta_{k}} \cap \mathcal{S}^{\cap \beta_{k*}} = \emptyset$, then the cost of the optimal offline solution increases by 1, which by Lemma~\ref{lemma: H_m} irrevocably decreases the advantage of $\text{ADV}^B$. Therefore, suppose that $\text{ADV}^B$ reveals a batch $\beta_k$ such that $\mathcal{S}^{\cap \beta_{k}} \cap \mathcal{S}^{\cap \beta_{k*}} \neq \emptyset$. Moreover, with non-zero probability $\text{ALG}^B$ disregards all $S_j \in \mathcal{S}(\beta_k) \setminus \mathcal{S}^{\cap \beta_{k}}$, so safely assume this is the case for the rest of the proof. Now, with non-zero probability $\text{ALG}^B$ disregards all $S_j \in \mathcal{S}^{\cap \beta_{k}} \setminus \mathcal{S}^{\cap \beta_{k*}}$, if any, making such deviation futile. Thus, safely assume that $\mathcal{S}^{\cap \beta_{k}} \subset \mathcal{S}^{\cap \beta_{k*}}$, implying that $t < r$. Let $k' > k$ be the first step after step $k$ such that $\mathcal{S}^{\cap \beta_{k}} \cap \mathcal{S}(\beta_{k'}) \neq \emptyset$. As before, with non-zero probability $\text{ALG}^B$ disregards all $S_j \in \mathcal{S}(\beta_{k'}) \setminus \mathcal{S}^{\cap \beta_{k'}}$, so safely assume that this is the case for the rest of the proof. Then, with non-zero probability $\text{ALG}^B$ also disregards any $S_j \in \mathcal{S}^{\cap \beta_{k'}} \setminus \mathcal{S}^{\cap \beta_{k}}$, if any, so safely assume that $\mathcal{S}^{\cap \beta_{k'}} \subset \mathcal{S}^{\cap \beta_{k}}$. Then, by the inductive hypothesis, given that batch $\beta_k$ satisfied  $\mathcal{S}^{\cap \beta_{k}} \subset \mathcal{S}^{\cap \beta_{k*}}$, the best $\text{ADV}^B$ can do is to recreate $I_z^*(t)$ on the remainder of the steps, starting with $k'$. In particular, this requires $\mathcal{S}^{\cap \beta_{k'}} \subset \mathcal{S}^{\cap \beta_{k}}$ and $|\mathcal{S}^{\cap \beta_{k'}}| + 1 = |\mathcal{S}^{\cap \beta_{k}}|$. A symmetric argument can be made about concurrently recreating $I_z^*(r-t)$ on the remainder of the steps, which is disjoint from $I_z^*(t)$ after the $k$th step and hence increases the offline solution by one. However, by Lemma~\ref{lemma: H_m}, this achieves a strictly lower competitive advantage for $\text{ADV}^B$.

Note that $\text{ADV}^B$ cannot be guaranteed to benefit from producing batches $\beta_k$ satisfying $\text{VCD}(\beta_k, \mathcal{S}) > z$ because with non-zero probability $\text{ALG}^B$ disregards any $S_j \in \mathcal{S}(\beta_k) \setminus \mathcal{S}^{\cap \beta_{k}}$, making such deviation futile with non-zero probability. In fact, a larger VC-dimension would involve more sets, possibly making $t=|\mathcal{S}^{\cap \beta_k}|$ (and hence the attainable competitive advantage) smaller.

\end{proof}

\end{document}